\documentclass[notitlepage, a4paper, superscriptaddress, 10pt,tightenlines,twocolumn,pra,nofootinbib]{revtex4-2}
\usepackage[utf8]{inputenc}
\usepackage{amsmath,amssymb,amsthm,bm,mathtools,amsfonts,mathrsfs,bbm,dsfont}
\usepackage{physics}
\usepackage{tikz}
\usepackage{tabularx} 

\usepackage{chngcntr}
\newtheorem{proposition}{Proposition}

\usepackage{xcolor}

\definecolor{dark-gray}{rgb}{.35,.55,.55}
\definecolor{dark-blue}{rgb}{.0,.0,.6}
\usepackage[colorlinks=true,linkcolor=dark-blue,citecolor=dark-blue,urlcolor=dark-blue]{hyperref}
\usepackage[capitalise]{cleveref}

\newcommand{\pr}[2][]{ \mathbb{P}_{#1} \left(  #2  \right) }
\newcommand{\id}{ \mathbbm{1} }
\newcommand{\flux}{ \mathfrak{f} }

\newcommand{\down}{ _{\downarrow} }

\makeatletter
\def\maketitle{
\@author@finish
\title@column\titleblock@produce
\suppressfloats[t]}
\makeatother

\begin{document}

\title{Aging and Reliability of Quantum Networks}

\author{Lisa T. Weinbrenner} 
\email{lisa.weinbrenner@uni-siegen.de}
\affiliation{Naturwissenschaftlich-Technische Fakultät, 
Universität Siegen, Walter-Flex-Straße 3, 57068 Siegen, Germany}

\author{Lina Vandré}
\affiliation{Naturwissenschaftlich-Technische Fakultät, 
Universität Siegen, Walter-Flex-Straße 3, 57068 Siegen, Germany}
\affiliation{State Key Laboratory for Mesoscopic Physics, School of Physics and Frontiers Science Center for Nano-Optoelectronics, Peking University, Beijing 100871, China}

\author{Tim Coopmans}
\affiliation{Leiden Institute of Advanced Computer Science, Leiden University, Leiden, The Netherlands}

\author{Otfried Gühne}
\affiliation{Naturwissenschaftlich-Technische Fakultät, 
Universität Siegen, Walter-Flex-Straße 3, 57068 Siegen, Germany}

\begin{abstract}
Quantum information science may lead to technological breakthroughs in 
computing, cryptography and sensing. For the implementation of these 
tasks, however, complex devices with many components are needed 
and the quantum advantage may easily be spoiled by failure of few  
parts only. A paradigmatic example are quantum networks. There, 
not only noise sources like photon absorption or imperfect quantum 
memories lead to long waiting times and low fidelity, but also hardware 
components may break, leading to a dysfunctionality of the entire network.
For the successful long-term deployment of quantum networks in the future, 
it is important to take such deterioration effects into consideration during 
the design phase.
Using methods from reliability theory and the theory of aging we develop 
an analytical approach for characterizing the functionality of networks
under aging and repair mechanisms, also for non-trivial topologies. Combined 
with numerical simulations, our results allow to optimize 
long-distance entanglement {distribution} under aging effects.  
\end{abstract}

\date{\today}

\maketitle

\section{Introduction}

Modern experimental science is based on complex technical devices for 
measuring and manipulating physical systems. Examples are the Large 
Hadron Collider for testing the limits of the standard model and the 
observatory LIGO for detecting gravitational waves. All these experimental 
setups consist of many different parts, which must be functional for 
the whole device to be operational. In quantum technologies, typical 
setups include quantum computers or quantum networks which consist of 
many smaller quantum systems. These constituents are not only prone to 
decoherence and noise, but they may also fail completely.  
For overcoming this, redundancy 
and quantum error correction can be used and there is a trade-off
between the required redundancy and the quality of the elementary devices. 
This trade-off is also relevant for comparing different quantum computing 
platforms, such as solid-state systems with many qubits and relatively
short coherence times with ion traps, which have fewer but long-lived
qubits \cite{linke2017, blinov2021comparison}.
For the long-term success, developing different approaches towards quantum technologies is important, but an advanced theory for analyzing the pros 
and cons of different implementation strategies is needed.

Besides quantum computers, quantum networks are a central paradigm of quantum 
technologies \cite{Kimble2008, Wehner2018}. The aim of these networks is to enable 
global quantum communication, but they are also useful for distributed sensing \cite{Proctor2018, Guo2020, Sekatski2020}, 
clock synchronisation \cite{Komar2014} and blind 
quantum computation \cite{Barz2012}. Consequently, there are many theoretical proposals for 
robust and efficient networks, using techniques like quantum repeaters \cite{briegel1998quantum, Duan2001, Sangouard2011, azuma2022quantum}, multiplexing \cite{Collins2007, munro2010quantum,sinclair2014spectral,dhara2021subexponential} and advanced quantum state encodings \cite{azuma2015, borregaard2020oneway, muralidharan2016optimal}.
In the last years, different building 
blocks of quantum networks were experimentally demonstrated, such as quantum 
teleportation between non-neighboring nodes \cite{2022Natur.605..663H}, satellite to ground 
communication \cite{Yin2020} and multiplexing 
\cite{2017NatCo...815359P, sperling2015}. All these implementations are, apart from the ubiquitous noise and photon loss, subject to 
deteriorating effects: devices may be defect from the beginning or get 
destroyed during the experiment. In practice, devices may also fail only temporarily, due to, for example, overheating.
While noise {and loss} in networks have been discussed
in detail
and may be overcome by entanglement purification, error correction {and other heralded procedures~\cite{muralidharan2016optimal}},
the effects of aging, breaking and repairing of {hardware} components are rarely studied. 
The present literature closest to this topic focuses on the waiting times 
{and quantum-state quality} for the functionality of a chain \cite{Azuma2021,Collins2007, praxmeyer2013reposition, Shchukin2019, kamin2022exact, vinay2019statistical, dai2021entanglement,brand2020efficient,li2021,coopmans2022improved,inesta2023optimal}, 
where temporary {hardware} failure and recovery are absent, or global success probabilities for networks 
\cite{Khatri2019, bugalho2023distributing}, where the topological structure of a specific network is often not taken into account.
For a full characterization 
of aging effects and reliability, three main difficulties arise: First, the 
devices are highly dependent on each other due to feedback loops.
Second, 
the topologies of networks are inherently more complicated than simple 
repeater chains. Third, the temporary failures of single devices may 
lead to temporal dependencies of the working probability and correlations 
in time.

The goal of this paper is twofold: First, we show how ideas and concepts 
from reliability theory \cite{gnedenko2014mathematical, Gavrilov2001, bazovsky2004reliability}, known 
from sociology and the theory of aging, can be applied to the analysis 
of quantum networks. Second, we develop an analytical theory for 
discussing temporary failure and repairing in quantum networks. 
In contrast to classical networks, where adding redundant hardware is cheap, current quantum hardware is fragile so that it is 
essential to have complete understanding of the network 
behaviour to find the optimal protocols for it. 
Our approach differs from network percolation theory \cite{Newman2018, Albert2002, coutinho2022robustness} 
by focusing on the temporal failure and recovery 
and is in line with recent work on fabrication defects
in planar quantum computers \cite{strikis2021quantum}.
The presented methods lead to various results.
First, depending on a 
given figure of merit for the entire network, we can specify which type and quantity of 
the redundant devices like quantum memories and entanglement sources are 
needed.
Second, we can characterize the temporal correlations 
of a quantum network depending on the failure and repairing of each 
single device.
Third, our results lead to improved cut-off times for key-generation
protocols performed on these networks.
As a main example we use the suggested 
configuration for a real quantum network in the Netherlands \cite{rabbie2022designing}, but 
our results are easily extendable to more complex topologies.

The article is structured as follows: First, we explain in Sections~II and~III the necessary concepts of quantum networks, and the theory of aging and reliability theory. In Section~IV we then apply reliability theory to quantum networks in general and demonstrate this approach for two different network topologies. In Section~V we then introduce the concept of repairing components in quantum networks, and analyze the temporal correlations of broken components via different correlation measures in Section~VI. At last, we use our results in Section~VII to obtain improvements for quantum protocols performed on two different topologies of imperfect quantum devices.

\section{Quantum Networks}

It is the general aim of a quantum network to generate entanglement 
in two distinct, far apart stations. 
However, the loss probability when exchanging entangled particles through a glass fibre increases exponentially with the distance.
Therefore, one can use the concept of a quantum 
repeater \cite{briegel1998quantum, Duan2001, Sangouard2011}, where several intermediate repeater stations are built between the end stations 
(see Fig.~\ref{fig:repeaterchain}). Entanglement is first created 
and stored between adjacent repeater stations and then converted to 
long-range entanglement via Bell measurements within the stations. 
This setup allows to overcome the exponential loss scaling. More 
generally, a quantum network may consist of many nodes which are 
connected by edges in various ways, leading to a complicated 
topology. 

\begin{figure}[t]
\includegraphics[width=0.9\columnwidth]{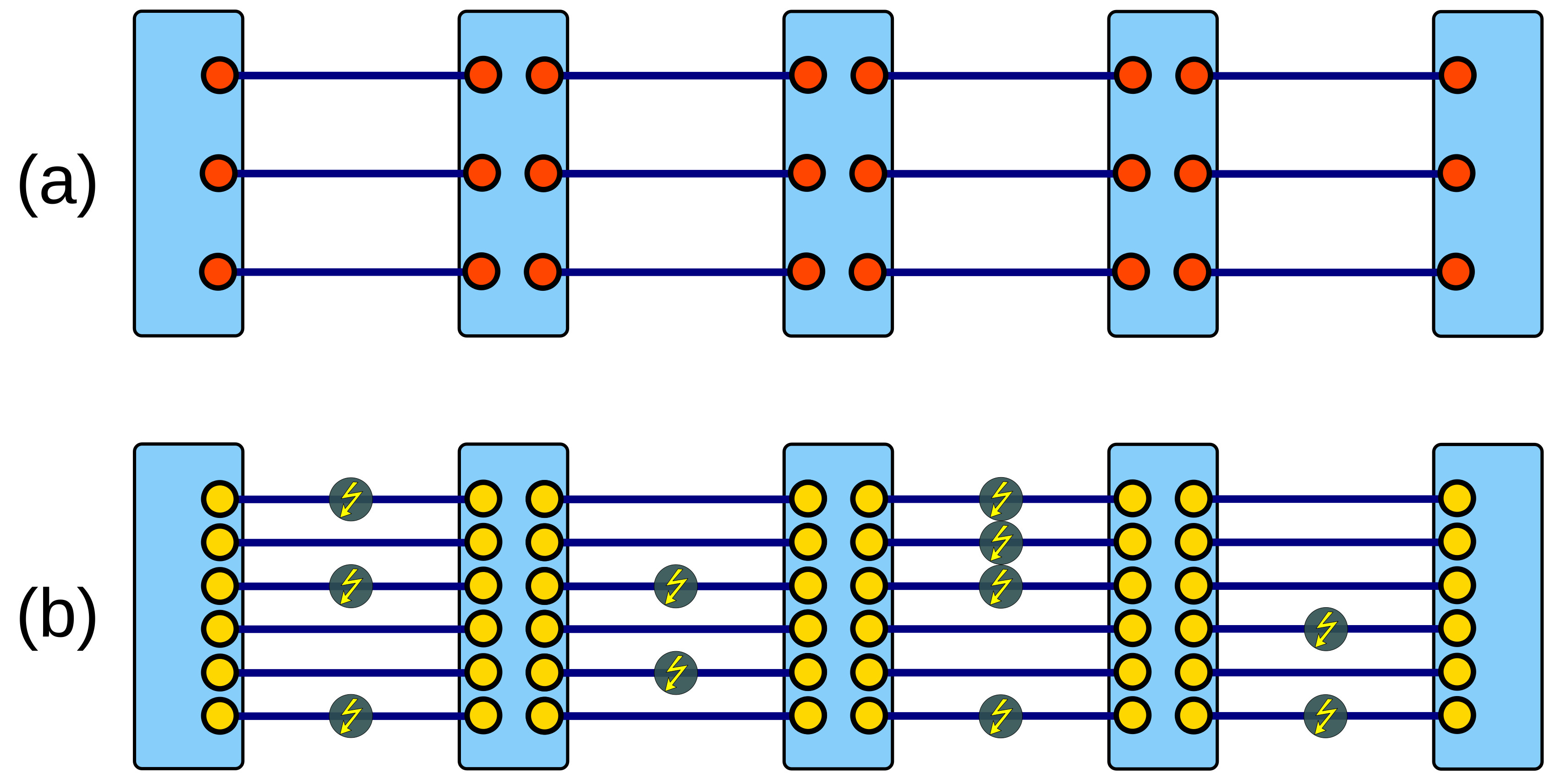}
\caption{(a) Schematic view of a multiplexed repeater chain, with 
$M+1=5$ nodes and $M=4$ edges, each edge having a multiplicity of 
$N=3$. 
(b) Alternative approach for the implementation of a similar 
repeater chain, where initially $N=6$ physical connections have 
been established within the edges, but some of them are not 
functioning from the beginning. In this example, on average 
three physical connections are working. In the depicted example, 
there are two connections between the end nodes (as one edge 
contains only two working connections), so the flux of the 
entire chain is two. Interestingly, both approaches, although 
being initially on average the same, can exhibit fundamentally 
different behaviour on the long time scale.
\label{fig:repeaterchain}
}
\end{figure}

In order to explain our main ideas, we first focus on a simple network, the 
so-called repeater chain, consisting of $M+1$ nodes connected with $M$ edges 
in a linear configuration, see also Fig.~\ref{fig:repeaterchain}. We allow  
multiplexing within the edges, meaning that each edge consists of $N$ physical 
connections. Each connection can be seen as a physical system, where entanglement 
may be established, e.g., a spatial or spectral mode of a fiber for 
photons. The concept of multiplexing was introduced to allow for higher 
entanglement generation rates \cite{Collins2007}.

Our main aim is to characterize at which times the network is functional 
if some devices or connections break down permanently or become dysfunctional 
for a certain time, e.g., due to ice formation on a chip in a cryostat, overheating, 
or mechanical failures. So, instead of describing 
the entanglement generation process and its time scales (as it has been 
done in Refs.~\cite{Khatri2019, Collins2007, Shchukin2019}, 
for example), we want 
to describe the functionality of the entire network depending on the 
functioning of each single device. 

\section{The theory of aging}

Biological and technical systems typically consist of many parts, with a complex 
structure of functional dependencies. This raises the question how long the entire 
system is functional, if some of the parts fail.  The description of the reliability 
of such a complex system is the main goal of reliability theory \cite{gnedenko2014mathematical,bazovsky2004reliability}. The two main quantities are the reliability function 
\begin{align}\label{eq:reliability}
S(t)=\pr{T>t},
\end{align}
which denotes the probability for a failure to occur at the failure time $T$ 
{\it after} time $t$, so it describes the probability for the system to work at 
least until time $t$. This directly allows to compute the mean time to 
failure via $\langle T \rangle = \int \!dt S(t)$. 
The failure rate
\begin{align}
\label{eq:failurerate}
\mu(t)= - \frac{\partial_t S}{S} = -\partial_t\ln S(t)
\end{align}
describes the probability of a failure in the next time interval 
given the survival of the system until time $t$. Characteristic for 
the phenomenon of aging is the fact that the failure rate increases 
with time. 

Starting from these two quantities, reliability theory offers a vast
spectrum of results and algorithms. The reliability of a complicated system
can, e.g., be calculated by a connection to repairable flow networks, for which 
many efficient algorithms are known \cite{mt2013flow}. Another main point
in reliability theory is the effect of maintenance on the reliability of 
a system, considering pre-agreed strategies as well as adaptive strategies
\cite{gertsbakh2000reliability}.

The starting point of our analysis arises from the suggestion of L.~Gavrilov and 
N.~Gavrilova to explain the different observed failure rates of technological devices
(given by the Weibull power law $ \mu(t)=a t^b $ \cite{weibull1939statistical}) and 
biological systems (growing exponentially according to the Gompertz-Makeham 
law $\mu(t)=A+B e^{\lambda t}$ \cite{makeham1860law, gompertz1825nature}) by using
two different notions of redundancy, see Fig.~\ref{fig:repeaterchain}. In technical 
systems (e.g., an aircraft), each sub-component (e.g., the measurement devices of 
the air speed) appears in a fixed degree of redundancy and can be assumed to work
in the beginning due to supervision and testing. In biological systems, however, 
the sub-components (e.g., the organs) have a varying degree of redundancy, since 
some of their cells may not work right from the beginning. See Ref.~\cite{steinsaltz2004markov} 
for a criticism of the conclusions from these models in Refs.~\cite{gavrilov1991quantitative, Gavrilov2001}.

\section{Aging theory for chains and networks}

\begin{figure}[t]
\includegraphics[width=\columnwidth]{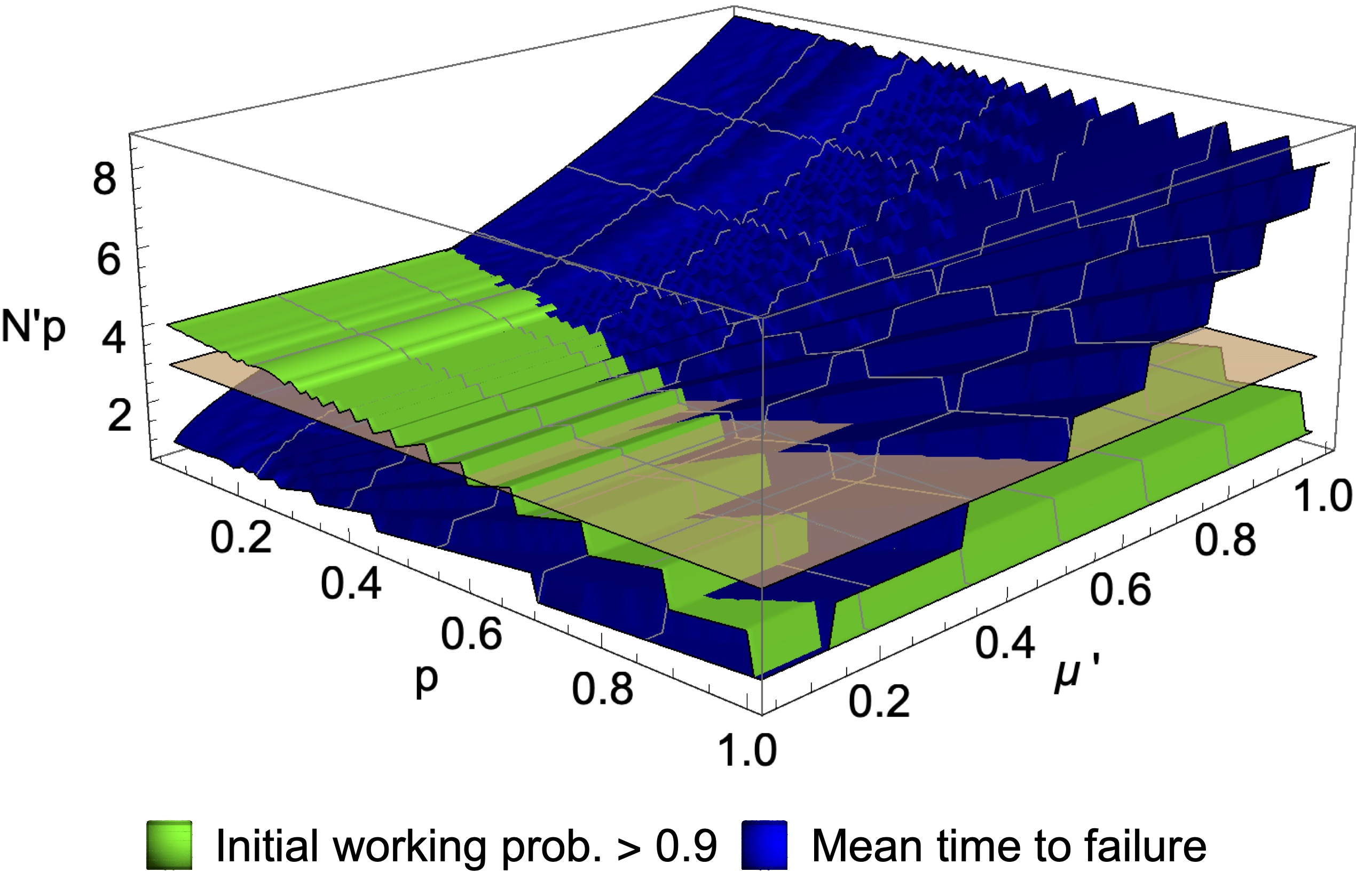}
\caption{Comparison between the implementations (a) and (b) [see Fig.~\ref{fig:repeaterchain}] with respect to two different figures of merit: Minimal multiplicity $N'p$ required in implementation (b) to reach the desired functionality in terms of life time and initial working probability as in implementation (a) with $N=3$ and $\mu=1/2$. The constant red half-transparent surface is given by $N'p=3$. Note that the green surface given by the initial condition is constant with respect to $\mu'$. The sawtooth-like pattern of steps is a direct consequence of the fact that $N'$ is an integer number.
See text for more details.
\label{fig:chainresult}
}
\end{figure}

The methods from Ref.~\cite{Gavrilov2001} can directly be extended to study 
the repeater chain. As implementation (a), corresponding to 
Fig.~\ref{fig:repeaterchain}(a), we 
consider $M$ edges with multiplicity $N$, working perfectly in the beginning, 
but affected by some failure rate $\mu$, which describes the probability 
of a single edge to break down in the next time interval.
As implementation (b), corresponding to Fig.~\ref{fig:repeaterchain}(b), we consider the $M$ edges 
connected with a higher multiplicity $N'>N$, where, however, with a certain 
probability $1-p$ each connection is not working in the beginning. In practice, this 
may be caused by a lower quality of the devices, which are still favourable 
due to a lower price. In this case, the devices are affected 
by some failure rate $\mu'$.
For these two implementations, we are interested in the total number of parallel 
connections between the end nodes; this so-called flux corresponds to the 
minimum of the working connections taken over all edges. The flux describes the capacity for communication from one point to the other; it is 
therefore a crucial quantity in the real implementation of a network to 
analyze, e.g.,
the transmission delays for quantum communication \cite{diadamo2022packet,mandil2023quantum}. Extending the 
methods from Ref.~\cite{Gavrilov2001}, any of the mentioned quantities can be 
determined fully analytically, details are given in Appendix~A. 

Naively, one may expect that the two implementations behave similarly, 
if $N'p =N$ and $\mu' = \mu$, that is, the  number of initially working 
connections in the second implementation equals {\it on average} the 
number in the first implementation. This is, however, not the case. 
As a concrete example, we consider implementation (a) with $M=6$, $N=3$ and 
$\mu=1/2$  and implementation (b) with
$M=6$, varying $p$ 
and varying $\mu'$. Then, we ask how large $N'$ in implementation (b) needs 
to be such that implementation (b) has the same mean time to failure
$\langle T \rangle$ as implementation (a).
This is depicted as blue surface in Fig.~\ref{fig:chainresult} 
and one finds that $N'p$ is typically significantly larger than $N$, even 
for $\mu'=\mu=1/2.$

Furthermore, in implementation (b) it is not guaranteed at all 
that the chain works from the beginning. So we ask (for the 
same parameters) how large $N'p$ needs to be to guarantee that 
implementation (b) has in the beginning a working probability
of at least $p_{\text{thres}}=0.9$ (green surface in 
Fig.~\ref{fig:chainresult}).
Clearly, the lower the probability 
that a single connection is functioning, the higher multiplicity 
is needed for an initially working chain. Interestingly, for 
much lower failure rates $\mu'<\mu$ this is also the determining 
factor for $N'$, if both conditions shall be satisfied.
Hence we arrive at two different effects, relevant in different 
parameter regimes, that can be useful in determining which devices 
to consider when building a new network. 
Cheaper and imperfect 
devices with the same failure rate may seem tempting if the price 
of $N'=N/p$ devices compared to the price of $N$ ``perfect'' devices 
is lower. However, as derived above, one actually needs a much 
higher number of cheaper devices than one would naively think.

While the repeater chain may be considered as a simple toy 
model, future quantum networks will have complex 
topologies \cite{rabbie2022designing}. Similarly, the components 
of other quantum technological devices are likely to exhibit 
more complicated functional dependencies than the ones underlying 
the repeater chain. In order to demonstrate that our methods 
are capable of dealing with these, we consider a network as in 
Fig.~\ref{fig:simple} (b), which can be seen as a structural approximation
of an optimized network under realistic conditions \cite{rabbie2022designing}. 
In Appendix~B we derive a formalism to deal with this type of topologies 
using the method of indicator functions instead of probabilities.

\begin{figure}[t]
\includegraphics[width=0.9\columnwidth]{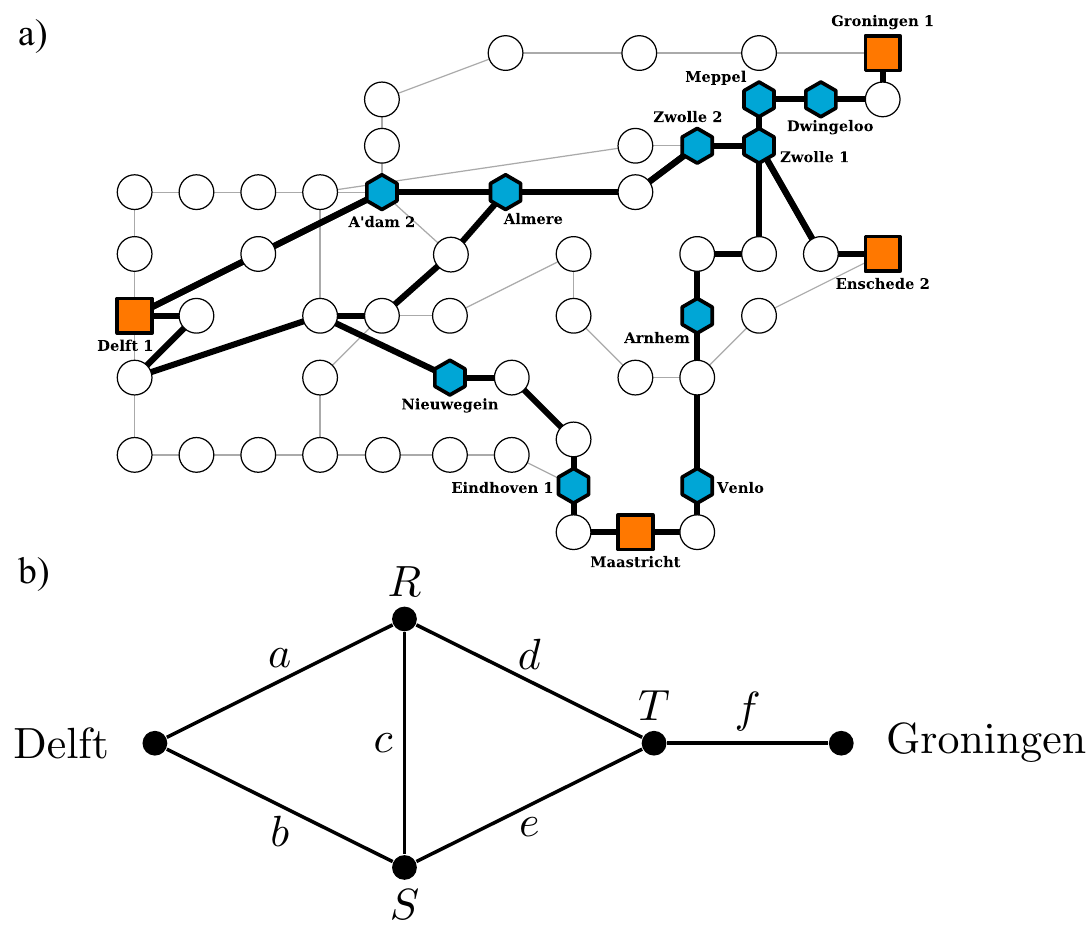}
\caption{(a) Optimized network topology for a network relying on 
the the Dutch telecom infrastructure (figure taken from Ref.~\cite{rabbie2022designing}). (b) Structural approximation of 
this network for a connection from Delft to Groningen. This network 
consists of five nodes and six connections.
\label{fig:simple}
}
\end{figure}


We now demonstrate in some detail the methods described above on two different examples, one of them being the topology of Fig.~\ref{fig:simple}.
To summarize briefly, we saw that there is a difference between the following two models:
\begin{itemize}
    \item[(a)] Each edge has a multiplicity of $N$, where every physical connection works perfectly in the beginning; however, the connections are affected by an exponential decay with failure rate $\mu$.
    \item[(b)] Each edge has a multiplicity of $N'$, where each physical connection is not working in the beginning with a certain probability $1-p$; the connections are also affected by an exponential decay with failure rate $\mu'$.
\end{itemize}


\begin{figure}
    \centering
    \includegraphics[width=0.9\linewidth]{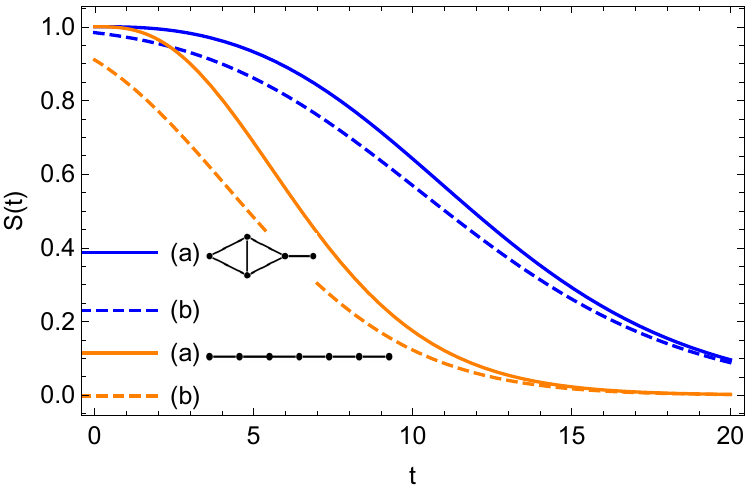}
    \includegraphics[width=0.9\linewidth]{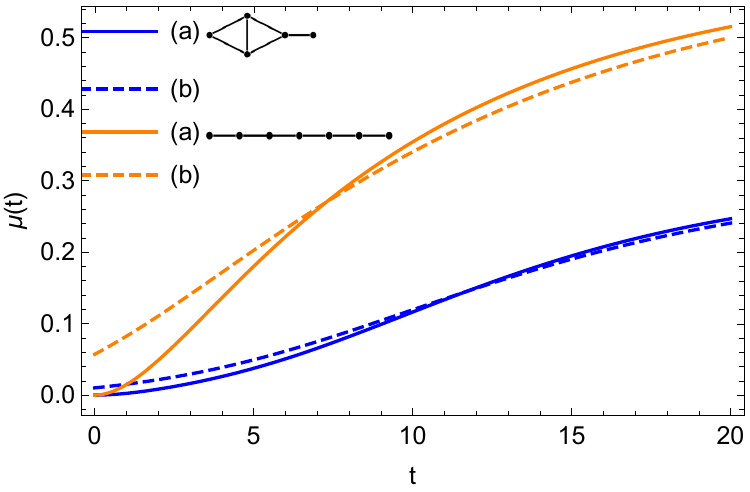}
\caption{Reliability function $S(t)$ and failure rate $\mu(t)$ of a network (blue or dark gray) or chain (orange or light gray) depending on the time $t$, both consisting of $6$ edges; the solid lines denote the results for model (a), where each edge has a multiplicity of $N=3$; the dashed lines denote the results for model (b), where each edge has a multiplicity of $N'=6$ and an initial failure probability of $1-p=1/2$.
For the description of the models see Fig.~\ref{fig:repeaterchain} and Sec.~IV.
\label{fig:rel_results}}
\end{figure}

In Fig.~\ref{fig:rel_results} we plot now the two resulting reliability functions and failure rates for these two models and two different topologies. The first topology is again a chain consisting of $M=6$ edges, compared to the simplified topology of the network in Fig.~\ref{fig:simple}, which also consists of $6$ edges. The results for the network are obtained using indicator functions as described in Appendix~B. There are several observations. First, the reliability functions according to model (b) have a negative offset at $t=0$. This should be expected, as there is a nonzero probability that the whole system is not functional right from the start as each connection is only functional with probability $p$. However, the offset is smaller in the case of the network since there are more possible end-to-end paths. Analogously, the failure rates for model (b) do not start at $0$ for $t=0$. A second observation shows that the two models behave the same in the long run. 
If the individual connections start to fail, then at some point one cannot distinguish anymore between an initially broken connection and a connection which just broke in the last few seconds.
However, in the short time range model (a) performs better. The last observation concerns the different topologies. It is not surprising that the network shows a more stable behavior, that is, a slower decreasing reliability and a slower increasing failure rate. The chain will fail as soon as one of the edges fails, the network however is robust to the failure of up to three edges as long as there exits still a path from end to end.

\section{Repairing components in networks}

Here, we consider multiplexed 
networks or repeater chains, where the physical connections 
may break down according to the models discussed above, but 
broken connections can be repaired. We assume that the 
repairing process takes several time steps. 

Several questions can be asked in this situation. The first one 
concerns the probability that, at a given point in time, the 
network is functional. 
Second, one can ask for the probability that 
the entire network is broken for several consecutive time steps. 
Formally, this is related to the calculation of waiting times
for entanglement generation in networks with finite memory 
\cite{Collins2007, praxmeyer2013reposition} and is challenging 
due to the complicated temporal dependencies of the model. In the following, 
however, we will develop a method to tackle this.

We use for our calculations the following discrete instead of continuous model: Each connection 
can break in a single time step with constant probability $p\down$, so 
the breaking of the connection is geometrically distributed, which can be seen as the
discrete analogue of the exponential distribution with the mean value $1/p\down$.
If the connection now breaks, it remains non-functional for exactly 
$\tau$ time steps, where $\tau$ is a constant set in advance. 
After this, it is functional with probability 
$(1-p\down)$, but with probability $p\down$ it breaks directly 
again, leaving it broken for at least $2\tau$ consecutive time 
steps. 

Since the expectation value of a geometric distribution is given 
by $1/ p\down$, the connection is on average functional for 
$1/p\down -1$ time steps before breaking. The connection is 
thus broken on $\tau$ out of on average $1/p\down -1 + \tau$ 
time steps, so the average probability of a single connection to be broken
is given by $\tau/( 1/p\down -1 +\tau)$. For a repeater chain 
with $M$ edges of multiplicity $N$ this leads to an average 
working probability of
\begin{align}
q =  
\Big[1 -\Big(\frac{\tau}{\frac{1}{p_{\downarrow}} -1 + \tau}\Big)^N \Big]^M. 
\end{align}

By using more subtle counting arguments we can also calculate 
the probability that a chain or network displays a certain 
behavior in $t$ consecutive, but arbitrarily chosen time 
steps. For the exact calculations we refer to Appendix~C. 
Our results can be used to characterize the time dependencies
of the system and, in particular, to calculate the probability 
for a given network to be functional conditioned on the time step(s)
before.

\section{Correlation measures}

Here, we want to give an example how the results described above can be used to analyze the temporal correlations in various network structures. In Fig.~\ref{fig:repairingresult} we present in detail
the behavior in up to three consecutive time steps of 
a multiplexed network and a multiplexed repeater 
chain with repairing. Here, the chain consists again of $M=6$ edges, and for both 
topologies we have a multiplicity of $N=3$ and a repairing time 
of $\tau=7$. 
The total state of each of the networks can assume two values, $S=0$ for ``not working'' and $S=1$ for ``working''. The analytical results (see Appendix~C) for this random variable allow to calculate various correlations at different times. 

We first consider the normalized temporal correlation 
\begin{align}
    {\rm Cor}(t_a, t_b)= \frac{\langle{S(t_a)S(t_b)}\rangle - \langle{S(t_a)}\rangle \langle{S(t_b)}\rangle}{\sigma(S(t_a))\sigma(S(t_b))}
\end{align} for two consecutive 
time steps ($t_1, t_2$) or with one time step in between ($t_1, t_3$); 
where $\sigma(S(t_x))$ denotes the standard deviation.

\begin{figure}
\includegraphics[width=0.9\linewidth]{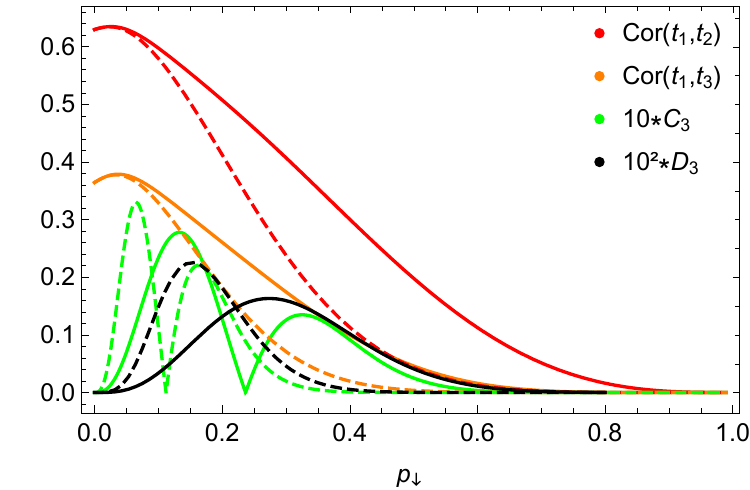}
\caption{Four different correlation measures demonstrating 
the dependencies in the temporal behaviour of two network 
models. Solid lines show the results for the network in 
Fig.~\ref{fig:simple}, the dashed lines for a repeater chain
similar to Fig.~\ref{fig:repeaterchain}(a) with $M=6$ connections. Both network structures have a multiplicity of $N=3$ and a repairing time of $\tau=7$.
\label{fig:repairingresult}
}
\end{figure}

Then, we consider two established measures of genuine tripartite correlation: 
The first one is the joint cumulant $C_3$ \cite{zhou2006cumulant}, which for three random
variables $A,B,C$ is defined as 
\begin{align*}
C_3 &= |\langle A B C\rangle 
-
\langle A \rangle \langle B C\rangle -
\langle B \rangle \langle A  C\rangle -
\langle C\rangle \langle A B \rangle \\
& \quad +
2 \langle A \rangle \langle B \rangle \langle C\rangle| .
\end{align*}
If the joint cumulant is nonzero, then none of the three random variables is 
independent from the other two. However, the other direction does not always hold.

The second measure $D_3$ of tripartite
correlations is based on exponential families and an extension of
the multi-information, which is used 
in the analysis of complex systems \cite{PhysRevE.79.026201, PhysRevE.85.046209, ay2017information}.  The Kullback-Leibler 
distance
\begin{align}
    D( P || Q) = \sum_{j} p_j \log({p_j}/{q_j})
\end{align}
describes how surprising the probability distribution $P$ is, 
if one has expected the distribution $Q$. Then, the correlation measure $D_3$
of a three-variable distribution is given by the minimal distance 
to all probability distributions $Q \in \mathcal{E}_2$ which 
are thermal states of two-body Hamiltonians, that is
\begin{align*}
D_3(P) = \inf_{Q\in\mathcal{E}_2} D(P||Q) .
\end{align*}
This describes genuine three-party correlations and is used 
in the analysis of complex systems \cite{PhysRevE.79.026201}.

Looking at the results in Fig.~\ref{fig:repairingresult} the 
first observation is that the repeater chain shows high time 
dependencies for smaller breaking probabilities $p\down$ than
the network. This is somewhat intuitive: If a complete failure 
of a system rarely but at least sometimes happens, then it 
is more probable that the system is broken for a few time steps 
in sequence than in randomly chosen time steps, leading to high 
time dependencies. Since the network is more robust to breaking, 
this behavior occurs for the network for higher breaking probabilities.
Second, the joint cumulant $C_3$ tends to zero if the 
system is equally often broken as functional. Around this 
point the correlation measure $D_3$ reaches its maximum. 
Third, every correlation measure tends to $0$ for high $p\down$, 
since in that case the systems remain broken most of the time. 

\section{Applications in entanglement distribution}

\begin{figure*}[t]
    \centering
    \includegraphics[height=5.1cm]{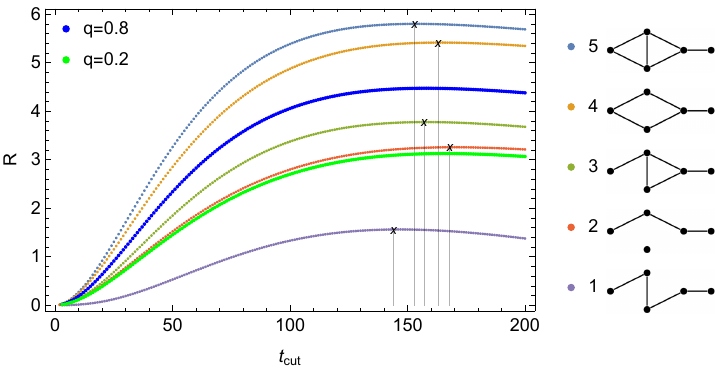}
    \includegraphics[height=5cm]{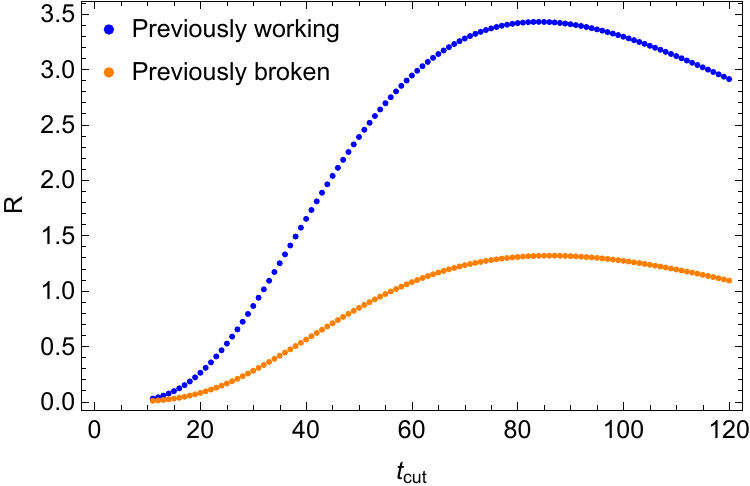}
    \caption{Secret-key rate (bits per second) in the simplified Netherlands network with $N=1$ (left) 
    and on a chain with $M=6$ edges with $N=3$ (right), depending on the chosen cut-off time. 
    See text for more details.
    \label{fig:rates}
    }
    %
    %
\end{figure*}

Now we want to use our results to derive 
improvements for quantum protocols run on these setups. 
We consider entanglement generation and subsequent quantum key distribution
with quantum repeaters as described above.
One problem in these schemes is that the quality of existing 
entangled links decreases when stored, due to decoherence. These weakly 
entangled links do not lead to high-fidelity long-range entanglement anymore
and block the generation of fresh entangled pairs, so it may be useful 
to restrict the duration of an entanglement generation attempt by a cut-off
time and erase all created links after this time~\cite{Collins2007,rozpedek2018parameters}. 
The cut-off time should be chosen in such a way that the secret-key rate achieved 
by this protocol is optimal~\cite{rozpedek2018near-term,kamin2022exact,li2021}.

We aim to find a strategy for optimizing the cutoff time for the network 
in Fig.~\ref{fig:simple} without multiplexing and for a chain of $M=6$ edges 
with multiplicity $N=3$. We assume two different
time scales; on the longer time scale devices break and are repaired, on top
of that and on a much shorter timescale 
entanglement is probabilistically distributed according to the protocol above. 
For the secret-key rates we now simulate the repeater protocol, assuming perfect Bell 
measurements and a entanglement generation probability of $P_{\text{gen}}=0.01$. 
Decoherence is modeled as depolarizing noise with a coherence time of $T_{\text{coh}} = 1s$ and the duration of one time step for the quantum protocol is $\frac{2}{3} 10^{-3} s$. For a motivation of 
these numbers and more details, see Ref.~\cite{avis2023asymmetric} and Appendix~D.

The results of our simulations are shown in Fig.~\ref{fig:rates}. In the left 
part the thin lines denote the secret-key rates achieved on the five different configurations of the Netherlands network 
depending on the chosen cut-off time. 
Assuming that a single edge of the network has an average working probability of 
$q=0.8$ or $q=0.2$ one can calculate the probabilities for the network to 
be in a certain configuration and from this the average secret-key rate achieved on the network. Clearly, these average key rates are 
lower than the one achieved on the complete network, so the 
classical defects of the technical devices have a significant impact 
on the quantum efficiency. A second observation is that the secret-key 
rate reaches a different maximum for each configuration, with the largest 
distance between two maxima given by 
$
\Delta t_{\text{cut}} = t_{\text{cut},2} - t_{\text{cut},1} = 168 - 144 = 24
$, 
corresponding to a relative change from configuration~$1$ to configuration~$2$
of around $17 \%$.
So, knowing the characteristics of the technical devices leads to a better 
choice of the cut-off time.

In the same way as above, we simulated the achievable key rates on a repeater chain with $M=6$ edges with multiplicity $N=3$. For the breaking probability we chose $p\down= 1/105$ and for the repairing time $s=15$. The resulting average key rates can be seen in the right part of Fig.~\ref{fig:rates}. Here, we took the state of the repeater chain in the previous time step into consideration and conditioned the probabilities for the different functional configurations on whether the repeater chain was functional or broken in the previous step.  One can see a clear difference between the key rates in these two different cases, which is to be expected. If the chain was functional in the time step before, then the configurations with a higher multiplicity in the edges are more probable than if the chain was broken in the previous step. So the knowledge about the previous behavior of the repeater chain can be used to choose an appropriate cut-off time for the current experiment and the breaking effects of the used technical devices can be mitigated by adapting the cut-off time of the quantum protocol to the current situation.

\section{Conclusion}

We developed analytical methods to treat failure and defects of classical
devices in quantum networks and which can be used to study temporal
correlations arising from failure-and-repair mechanisms. Based on our results, 
we showed that these classical effects have an impact on the performed quantum 
protocol and give rise to new criteria for an optimal cut-off time
in repeater protocols.
Our methods are also directly applicable to recent experimental implementations. One example is the four party network implemented in \cite{wengerowsky2018entanglement}, 
where different errors in the state generation or in the experimental 
equipment can directly be translated to failures of nodes or edges in our 
framework. Our methodology can also be applied to give additional 
constraints for developing new repeater architectures as in \cite{gu2024hybrid},
where a potentially multiplexed repeater chain is considered.

While our methods were developed in the quantum network paradigm, 
it seems promising to adapt them to other examples of quantum 
technologies. A concrete example are segmented ion traps 
 \cite{Hensinger_2006,Lekitsch_2017}. Here, 
ions are shuttled on a chip from one interaction zone to the other to
build quantum circuits with high fidelity. Failures of interaction
or shuttling procedures can directly be modeled with our approach 
and consequently the design of these ion traps can be optimized. On 
a more fundamental level, our results can be used to optimize 
entanglement distribution in networks \cite{bugalho2023distributing,Hansenne_2022} as well as 
multipartite cryptographic protocols \cite{Murta2020}. Finally, it would be interesting to exploit
the further results and methods from reliability theory \cite{mt2013flow, gertsbakh2000reliability,yeh2010particle, yeh2007interactive} to 
today's quantum technologies.

\section*{Acknowledgments}

We thank Jan L.~B\"onsel, Austin Collins, Antariksha Das, Sophie Egelhaaf,
Kian van der Enden, Tim Hebenstreit, 
Patrick Huber, Brian Kennedy, Peter van Loock,
and Fabian Zickgraf for discussions. This 
work was supported by the Deutsche Forschungsgemeinschaft  (DFG, German Research 
Foundation, project numbers 447948357 and 440958198), the 
Sino-German Center for Research Promotion (Project M-0294), 
the ERC (Consolidator Grant 683107/TempoQ), the German 
Ministry of Education and Research (Project QuKuK, BMBF Grant 
No. 16KIS1618K) and the Dutch National Growth Fund, as part of of the Quantum Delta NL program. L.V. thanks the Stiftung der Deutschen Wirtschaft and L.T.W. thanks the House of Young Talents of the University of Siegen.


\appendix

\section{Reliability Theory for Repeater Chains } 

In this Appendix we want to describe how one can apply reliability theory to repeater chains. First, we give a short introduction to the basics of reliability theory and the main quantities. Second, we derive some results for the easiest network structure, the repeater chain. We start by considering a single connection and then an edge or block, which consists of $N$ parallel connections. Here we differentiate between a block consisting of $N$ initially perfect connections and $N$ connections which may be broken in the beginning with probability $1-p$. Finally, these blocks of connections are combined to a chain. All the results are purely analytical. 

\subsection{Basics of Reliability Theory}

The \textit{reliability function} $S(t)$ describes the probability that a system fails after a certain time $t$ and is functional up to this point. If $T$ is the random variable describing the failure time of the system then
\begin{align}
	S(t) = \pr{T > t} .
\end{align}
Using the \textit{cumulative distribution function} 
\begin{align}
	F(t) = \pr{T \leq t}
\end{align}
one can express the reliability function as
\begin{align}
	S(t) = 1 - \pr{T \leq t} = 1 - F(t).
\end{align}
The \textit{failure rate} of the system is given by the logarithmic derivative of the reliability function
\begin{align}
	\mu(t) = - \frac{\dd}{\dd t} \ln{S(t)} = -\frac{1}{S(t)} \frac{\dd}{\dd t} S(t) = \frac{1}{S(t)} \frac{\dd}{\dd t} F(t). 
\end{align}
Sometimes it is easier to express the failure rate using the \textit{probability density function} 
\begin{align}
	f(t) = \frac{\dd}{\dd t} F(t) = - \frac{\dd}{\dd t} S(t).
\end{align}
It then holds
\begin{align}
	\mu(t) =  \frac{f(t)}{S(t)}.
\end{align}

\subsection{Results for Repeater Chains}

One block consists of a fixed number of $N$ parallel connections. Each of these connections can be functional or broken. The block has a \textit{flux} of $\flux$, if at least $\flux$ of the $N$ connections are functional. We want to calculate the probability to have a certain flux $\flux$ depending on the time $t$.
The calculation follows the strategy of Gavrilov and Gavrilova \cite{Gavrilov2001}, with some modifications to avoid the underlying implicit assumption of small time scales (see also Ref. \cite{steinsaltz2004markov}).
The important difference is that in their calculations a block fails if all connections are destroyed. In our calculation the block ``fails'', if $N-(\flux -1)$ connections are destroyed, because then a flux of $\flux$ is not possible anymore.

\subsubsection{A single connection}

A single connection has a fixed failure rate $\mu(t) = k= \mathrm{const}$, so it decays exponentially. Let $T$ be the random variable which describes the failure time of a single connection. For the exponential distribution it holds
\begin{align}
	S(1,k,t) := \pr{T>t} = 1- F(t) = e^{-kt} =: \alpha
\end{align}
and 
\begin{align}
	F(1,k,t) := \pr{T \leq t} = 1- e^{-kt}  = 1 - \alpha.
\end{align}

Note that it holds:
\begin{align}
	\frac{\dd \alpha}{\dd t}  = -k e^{-kt} = -k\alpha.
\end{align}

\subsubsection{A block with initially perfect connections}

One block consists of a fixed number of connections $N$. In the beginning all $N$ of these connections are functional.
We want to calculate the probability to have a certain flux $\flux$ depending on the time $t$.
\begin{proposition}\label{pro:blockPerfRel}
    Given a block of $N$ working connections which each fail according to a constant failure rate $\mu(t)=k$ the reliability function of the block $b$ for a flux $\flux$ at time $t$ is given by 
    \begin{align} 
    S_b^{\flux} (N,k,t) = \sum_{i=0}^{N-\flux} {N \choose i} \left[ 1- \alpha \right]^i  \alpha^{N-i}
    \end{align}
    and the failure rate is given by 
    \begin{align}
    \mu_b^{\flux} (N,k,t) = \frac{ k \flux {N \choose \flux} \left[ 1- \alpha \right]^{N-f}  \alpha^{\flux} }{ \sum_{i=0}^{N-\flux} {N \choose i} \left[ 1- \alpha \right]^i  \alpha^{N-i} } 
    \end{align}
    with $\alpha = e^{-kt}$.
\end{proposition}
\begin{proof}
Denote the failure time of the block with $T_b$ and the failure time of the $i$-th connection with $T_i$. The failure times of the connections are iid, so the probability for the failure of $i$ arbitrary connections is the same as for the connections $1,\dots,i$. The block fails in time $T_b\leq t$ and has therefore a flux $\leq \flux-1$, if at least $N-(\flux-1)$ of the connections fail before $t$:
\begin{align}
	& F_b^{\flux} (N,k,t) = \pr{T_b \leq t} \notag \\
	& \quad = \sum_{i=N-(\flux-1)}^{N} {N \choose i} \pr{T_1,\dots, T_i \leq t} \notag \\
    & \quad \phantom{\pr{T_1,\dots, T_i \leq t}} \times \pr{T_{i+1},\dots,T_N > t}  \notag \\
	& \quad = \sum_{i=N-(\flux-1)}^{N} {N \choose i} \pr{T_1\leq t}\cdots \pr{T_i\leq t} \notag \\
    & \quad \phantom{\pr{T_1,\dots, T_i \leq t}} \times \pr{T_{i+1} > t}\cdots \pr{T_{N} > t} \notag \\
	& \quad = \sum_{i=N-(\flux-1)}^{N} {N \choose i} \left[ F(1,k,t) \right]^i  \left[ S(1,k,t) \right]^{N-i} \notag \\
	& \quad = \sum_{i=N-(\flux-1)}^{N} {N \choose i} \left[ 1- \alpha \right]^i \alpha^{N-i} \notag \\
    & \quad = \sum_{i=0}^{N} {N \choose i} \left[ 1- \alpha \right]^i  \alpha^{N-i} -  \sum_{i=0}^{N-\flux} {N \choose i} \left[ 1- \alpha \right]^i  \alpha^{N-i} \notag \\
    & \quad = 1 -  \sum_{i=0}^{N-\flux} {N \choose i} \left[ 1- \alpha \right]^i  \alpha^{N-i}.
\end{align}
The reliability function is then given by
\begin{align}
	& S_b^{\flux} (N,k,t) = \pr{T_b > t} = 1- F_b^{\flux} (N,k,t) \notag \\
    & \quad = \sum_{i=0}^{N-\flux} {N \choose i} \left[ 1- \alpha \right]^i  \alpha^{N-i}. 
\end{align}
For the failure rate we calculate first the probability density function
\begin{align}
	& f_b^{\flux} (N,k,t) = -\frac{\dd}{\dd t} S_b^{\flux} (N,k,t) \\
	& \quad = -\frac{\dd \alpha}{\dd t} \frac{\dd }{\dd \alpha} \sum_{i=0}^{N-\flux} {N \choose i} \left[ 1- \alpha \right]^i  \alpha^{N-i} \notag \\
	& \quad = k\alpha  \big[ \sum_{i=0}^{N-\flux} {N \choose i} \left[ 1- \alpha \right]^i  (N-i) \alpha^{N-i-1} \notag \\
    & \qquad + \sum_{i=1}^{N-\flux} {N \choose i} (-1)i \left[ 1- \alpha \right]^{i-1}  \alpha^{N-i} \big] \notag \\
    & \quad = k\alpha  \big[ \sum_{i=0}^{N-\flux} {N \choose i} (N-i) \left[ 1- \alpha \right]^i  \alpha^{N-i-1} \notag \\
    & \qquad - \sum_{i=0}^{N-\flux-1} {N \choose i+1} (i+1) \left[ 1- \alpha \right]^{i}  \alpha^{N-(i+1)} \big] .
\end{align}
Then it holds 
\begin{align}
    {N \choose i+1} & (i+1) = \frac{N!}{(i+1)!(N-i-1)!} (i+1) \notag \\
    &= \frac{N!}{i!(N-i)!} (N-i) = {N \choose i} (N-i)
\end{align}
and therefore
\begin{align}
    f_b^{\flux} (N,k,t) &= k\alpha  \big[ \sum_{i=0}^{N-\flux} {N \choose i} (N-i) \left[ 1- \alpha \right]^i  \alpha^{N-i-1} \notag \\
    &\qquad - \sum_{i=0}^{N-\flux-1} {N \choose i} (N-i) \left[ 1- \alpha \right]^{i}  \alpha^{N-(i+1)} \big] \notag \\
    &= k\alpha  {N \choose N-\flux} \left[ N-(N-\flux) \right] \left[ 1- \alpha \right]^{N-\flux}  \alpha^{\flux-1} \notag \\
    &= k \flux {N \choose \flux} \left[ 1- \alpha \right]^{N-\flux}  \alpha^{\flux} .
\end{align}
The failure rate is then the fraction of the probability density function and the reliability:
\begin{align}\label{eq:rate,det,f}
	\mu_b^{\flux} (N,k,t) &= \frac{ f_b^{\flux} (N,k,t) }{ S_b^{\flux} (N,k,t) } \notag \\
	&= \frac{ k \flux {N \choose \flux} \left[ 1- \alpha \right]^{N-\flux}  \alpha^{\flux} }{ \sum_{i=0}^{N-\flux} {N \choose i} \left[ 1- \alpha \right]^i  \alpha^{N-i} } \notag \\
	&= \frac{ k \flux {N \choose \flux} \left[ \alpha^{-1}-1 \right]^{N-\flux} }{ \sum_{i=0}^{N-\flux} {N \choose i} \left[ \alpha^{-1} -1 \right]^i  }.
\end{align}
\end{proof}

\subsubsection{A block with probabilistic connections}

In this model one block consists of a fixed number of connections $N$. In the beginning, only $n$ of these connections are functional. The number $n$ depends on some probability distribution $q_n$ with $\sum_{n=0}^N q_n = 1$, e.~g., the binomial distribution $q_n = {N \choose n} p^n (1-p)^{N-n}$. The binomial distribution describes the case where each single connection is functional with probability $p$. In general the distribution $q_n$ can be arbitrary.

We want to calculate the probability to have a certain flux $\flux$ depending on the time $t$.

\begin{proposition}\label{pro:blockProbRel}
     Given a block of $N$ connections of which $n$ connections are working in the beginning with probability $q_n$ and denote the probability distribution by $Q=\{q_n\}$. Each working connection fails according to a constant failure rate $\mu(t)=k$. Then the reliability function of the block for a flux $\flux$ at time $t$ is given by
     \begin{align} 
        S_b^{\flux} (N, Q; k, t) = \sum_{n=\flux}^N q_n S_b^{\flux} (n,k,t) 
     \end{align}
     and the failure rate is given by 
     \begin{align}
        \mu_b^{\flux} (N, Q; k, t) = \frac{1}{S_b^{\flux} (N, Q; k, t)} \sum_{n=\flux}^N q_n \mu_b^{\flux} (n,k,t) S_b^{\flux}(n,k,t) .
     \end{align}
\end{proposition}
\begin{proof}
    Using the law of total probability it holds for the reliability
    \begin{align}
        S_b^{\flux} (N, Q; k, t) = \pr{T_b > t} = \sum_{n=\flux}^N q_n S_b^{\flux} (n,k,t).
    \end{align}
    For the failure rate it then follows
    \begin{align}
        & \mu_b^{\flux} (N, Q; k, t) = -\frac{1}{S_b^{\flux} (N, Q; k, t)} \frac{\dd}{\dd t} \sum_{n=\flux}^N q_n S_b^{\flux} (n,k,t) \notag \\
        & \quad = \frac{1}{S_b^{\flux} (N, Q; k, t)}  \sum_{n=\flux}^N q_n S_b^{\flux} (n,k,t) \notag \\
        & \quad \qquad \times \left[ -\frac{1}{S_b^{\flux} (n,k,t) }\frac{\dd}{\dd t} S_b^{\flux} (n,k,t) \right] \notag \\
        & \quad = \frac{1}{S_b^{\flux} (N, Q; k, t)}  \sum_{n=\flux}^N q_n S_b^{\flux} (n,k,t) \mu_b^{\flux} (n,k,t) .
    \end{align}
\end{proof}
In Ref. \cite{Gavrilov2001} they use at this point implicitly an approximation for small time scales by calculating the failure rate for a flux of $\flux =1$ without the weights $\frac{S_b^{1} (n,k,t)}{S_b^{1} (N, Q; k, t)}$:
\begin{align}
    \mu_b^{1} (N, Q; k, t) = \sum_{n=1}^N \tilde{q}_n \mu_b^{1} (n,k,t) .
\end{align}
The probabilities $\tilde{q}_n$ are normalized, such that the probability $\tilde{q}_0$ for a failure right in the beginning is $0$. In this case, the reliabilities $S_b^{1} (N, Q; k, 0) = 1$ and $S_b^{1} (n,k,0) = 1$ cancel for $t=0$ (since the device is surely functional in the beginning). However, for larger times $t$ the weights become relevant and lead to a very different behavior than in Ref. \cite{Gavrilov2001}.

\subsubsection{A chain}

\begin{proposition}\label{pro:chain}
    Given a chain of $M$ blocks with respective reliability functions $S_{b_j}(t)$ and failure rates $\mu_{b_j}(t)$, $j=1,\dots,M$, then the reliability function and failure rate of the whole chain are given by the product or sum, respectively:
    \begin{align}
        S_\text{chain} = \prod_{j=1}^M S_{b_j} (t) \quad \text{and}\quad \mu_\text{chain} (t) = \sum_{j=1}^M \mu_{b_j} (t) .
    \end{align}
\end{proposition}
\begin{proof}
Let $T_\text{chain}$  be the failure time of the chain and $T_{b_j}$ the failure time of block $j=1,\dots,M$. The failure times of the blocks are independent, but not necessarily identically distributed. The whole chain fails if at least one of the blocks fails. Therefore it holds
\begin{align}
	S_\text{chain} (t) &= \pr{T_\text{chain} > t} = \pr{T_{b_1},\dots, T_{b_M} > t} \notag \\
	&= \prod_{j=1}^M \pr{T_{b_j} > t} = \prod_{j=1}^M S_{b_j} (t)  
\end{align}
and
\begin{align}
	\mu_\text{chain} (t) &= - \frac{\dd}{\dd t} \ln{S_\text{chain} (t)}
	= - \frac{\dd}{\dd t} \ln{\prod_{j=1}^M S_{b_j} (t) } \notag \\
	&= \sum_{j=1}^M \left( - \frac{\dd}{\dd t} \ln{ S_{b_j} (t) } \right)
	= \sum_{j=1}^M \mu_{b_j} (t) .
\end{align}
\end{proof}

If the blocks are all iid, then it simply holds
\begin{align}
	S_\text{chain} (t) = \left[ S_{b} (t) \right]^M  
\end{align}
and 
\begin{align}
	\mu_\text{chain} (t) = M \mu_{b} (t) .
\end{align}


\section{Probabilities in more complex topologies}\label{ch:repairing}

To calculate the probabilities for a network to be functional or not, one has to consider the different possible paths in this network that lead to a resulting connection from one point of the network to another. This type of probabilities can e.g. be computed using the principle of inclusion and exclusion. Another easier and more clear option is the use of indicator functions. This idea can,
e.g., be found in Ref.~\cite{engel1973wahrscheinlichkeitsrechnung}.
We use indicator functions of the type
\begin{align}
    \id_X = \begin{cases}
        1, & \text{if X is functional} ,\\
        0, & \text{if X is broken} .
    \end{cases}.
\end{align}
Then it holds 
\begin{align}
    \mathbb{E}[ \id_X ]  &= 1 \times \pr{X\text{ is functional}} + 0\times \pr{X\text{ is broken}} \notag \\
    &= \pr{X\text{ is functional}}.
\end{align}
So to calculate the probability for the whole system to be functional we first calculate the indicator function and then take the average value.

If $A$ and $B$ are two components, which are connected in series to get the larger component $X$, then $X$ is functional iff $A$ and $B$ are functional. The indicator function of $X$ then reads as
\begin{align}
    \id_X &= \begin{cases}
        1, & \text{if } \id_A = \id_B =1 ,\\
        0, & \text{else}
    \end{cases} \notag \\ 
    &= \id_A \times \id_B
\end{align}
and the probability is therefore (since $A$ and $B$ are independent)
\begin{align}
    & \pr{X\text{ is functional}} = \mathbb{E}[ \id_X ] = \mathbb{E}[ \id_A ]\times \mathbb{E}[ \id_B ] \notag \\
    & \quad = \pr{A \text{ is functional}}\times \pr{B\text{ is functional}}.
\end{align}

If the components are connected in parallel, the indicator function of $X$ instead reads
\begin{align}
    \id_X &= \begin{cases}
        1, & \text{if } \id_A =1 \text{ or } \id_B =1 ,\\
        0, & \text{else}
    \end{cases} \notag \\
    &= 1- (1-\id_A) \times (1-\id_B)
\end{align}
and the probability is (since $A$ and $B$ are independent)
\begin{align}
    & \pr{X\text{ is functional}} = \mathbb{E}[ \id_X ] \notag \\
    & \quad = 1- (1-\mathbb{E}[ \id_A ])\times (1-\mathbb{E}[ \id_B ]) \notag \\
    & \quad = 1- \pr{A \text{ is broken}}\times \pr{B\text{ is broken}}.
\end{align}

The idea to use indicator functions instead of probabilities becomes more interesting when considering more difficult networks. 
\begin{figure}
    \centering
    \includegraphics[width=0.9\columnwidth]{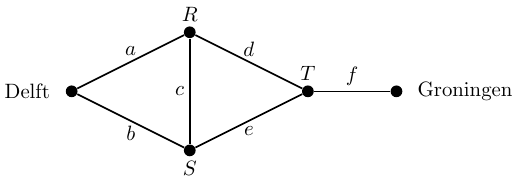}
    \caption{Simplification of the network in Ref. \cite{rabbie2022designing} from Delft to Groningen; the network consists of three components in series: a square network (with one diagonal connection $c$) from Delft up to node $T$, the node $T$ itself and a single connection $f$ to Groningen.}
    \label{fig:network}
\end{figure}
For the simple square network with one diagonal connection (see picture \ref{fig:network}; the square contains the connections from Delft to the node $T$, including nodes $R$ and $S$) the indicator function becomes
\begin{align}
    \id_{\text{square}} &= \begin{cases}
        1, & \text{if } \id_{aRd} =1 \text{ or } \id_{aRcSe} =1 \\
        & \text{ or } \id_{bSe} =1 \text{ or } \id_{bScRd} =1,\\
        0, & \text{else}
    \end{cases} \notag \\ 
    &= 1- (1-\id_{aRd}) \times (1-\id_{aRcSe}) \notag \\
    & \qquad \times (1-\id_{bSe}) \times (1-\id_{bScRd}),
\end{align}
where $\id_{aRd}$ is the short notation for $\id_{a} \id_R \id_d$. Using the fact that $\id^2 = \id$, one obtains
\begin{align}
    \id_{\text{square}} &= \id_{aR}(\id_{d}+\id_{cSe} - \id_{decS}) \notag \\
    & \quad + \id_{bS}(\id_{e} + \id_{cRd} - \id_{ecdR}) \notag \\ 
    & \quad - \id_{abRS} (\id_{de} + \id_{cd} + \id_{ce}-2\times \id_{ced}) .
\end{align}
Note that all the indicator functions in the products are independent. Therefore one can just calculate the average value of each single term to get the average value of $\id_{\text{square}}$. In the special case that the edges and nodes all have the same respective probabilities for failures $\pr{\text{edge is broken}} = e$ and $\pr{\text{node is broken}} = n$, the probability for the network to be functional is simply
\begin{align}
    & \pr{\text{square is functional}} = \mathbb{E}[\id_{\text{square}}] \notag \\ 
    & \quad = en(e+e^2n - e^3n) + en(e + e^2n - e^3n) \notag \\ 
    & \quad \quad - e^2n^2(e^2 + e^2 + e^2 -2\times e^3) \notag \\ 
    & \quad = 2e^2 n + e^3 n^2 (2-5e+2e^2).
\end{align}

The network in Ref. \cite{rabbie2022designing} from Delft to Groningen is then a series of a square network, a node $T$ and a single connection $f$. Therefore the indicator function for the connection from Delft to Groningen reads as
\begin{align}
    \id_X &= \id_{\text{square}} \times \id_{T} \times \id_{f} .
\end{align}


\section{Repairing of broken devices and resulting non-markovian effects}

In the following chapter we want to model the behavior of a network in which the connections not only break after some time, but also get repaired. We assume that the repairing takes some time. Mathematically speaking, the connections can break in each time step with probability $p\down$, e.g. the breaking is geometrically distributed. If a connection is broken it gets repaired after $\tau$ time steps. After this repairing time, the repaired connection can either break directly again with probability $p\down$ or stay functional with probability $1-p\down$. Note that this means that it is not guaranteed that after $\tau$ time steps the 
connection is functional again, it can well happen that it is broken for $2\tau$
or even more multiples of $\tau$.

We want to calculate the probability that a block of $N$ connections or a network behaves in a certain way in $t$ arbitrarily chosen consecutive time steps. 
To reach this goal we first calculate the probabilities for a single connection to behave in a certain way in up to three consecutive time steps. Then we show how these probabilities can be used to describe the behavior of a block of $N$ connections and more complex networks.

\subsection{A single connection}

We start be calculating all needed probabilities for a single connection. 

\begin{proposition}\label{pro:eff_prob}
    The probability for a single connection with failure probability $p\down$ and repairing time $\tau$ to be broken in an arbitrarily chosen time step is given by
    \begin{align}\label{eq:eff_prob}
        p_{\mathrm{eff}}^b = \frac{\tau}{1/p\down - 1 +\tau} .
    \end{align}
\end{proposition}
\begin{proof}
    Since the breaking of a connection is geometrically distributed, the first break will happen on average in time step $1/p\down$. Thus, the connection is on average functional for $1/p\down-1$ time steps and then broken for another $\tau$ time steps. The probability that the connection is broken in an arbitrarily chosen time step is thus
    \begin{align}
        p_{\text{eff}}^b = \frac{\tau}{1/p\down - 1 +\tau}
    \end{align}
    and the probability that the connection is functional is
    \begin{align}
        p_{\text{eff}}^f = \frac{1/p\down - 1}{ 1/p\down - 1 +\tau} = 1 - p_{\text{eff}}^b .
    \end{align}
\end{proof}
In the next step we want to calculate the probability that a single connection is broken in not only one arbitrarily chosen time step but instead in $1 \leq t \leq \tau$ consecutive arbitrarily chosen time steps. 
\begin{figure}
    \centering
    \includegraphics[width=0.9\columnwidth]{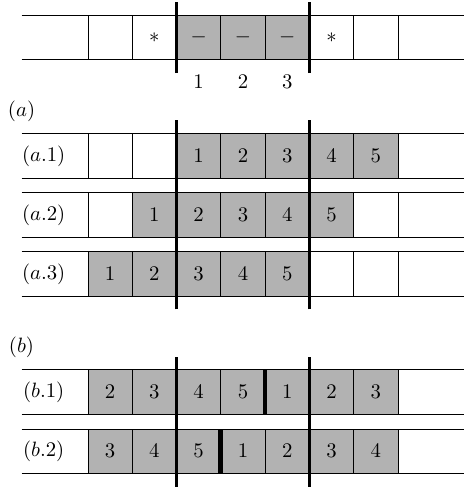}
    \caption{Possible states of a connection with repairing time $\tau=5$ to be broken in $t=3$ arbitrarily chosen consecutive time steps; $(a)$ if the first of the consecutive time steps coincides with one of the first $\tau - t+1= 3$ broken steps of the connection, the connection is broken throughout the three consecutive time steps; $(b)$ if the first of the consecutive time steps coincides with one of the last $t-1=2$ broken steps of the connection, then the connection has to break again directly after being repaired to remain broken through the three consecutive time steps.}
    \label{fig:prob_device_b}
\end{figure}
\begin{proposition}
    The probability for a single connection with failure probability $p\down$ and repairing time $\tau$ to be broken in $1 \leq t \leq \tau$ arbitrarily chosen consecutive time steps is given by
    \begin{align}
        m(t) = p_{\mathrm{eff}}^b \left( 1 - \frac{t-1}{\tau} (1 - p\down) \right) .
    \end{align}
\end{proposition}
\begin{proof}
    The probability for the connection to be broken in $t$ consecutive time steps $m(t)$ can be calculated in a similar way as in Proposition \ref{pro:eff_prob} (see also Fig.~\ref{fig:prob_device_b}). We saw there that a connection is on average functional for $1/p\down-1$ time steps and then broken for another $\tau$ time steps. Thus, the first of the $t$ consecutive time steps has to coincide with one of the $\tau$ broken steps of the device. If it coincides with one of the first $\tau-t+1$ broken steps of the device then the device remains surely broken for the following $t-1$ time steps. If it coincides with one of the $t-1$ last broken steps of the device then the device has to break again directly after being repaired to stay broken in the remaining $t-1$ consecutive time steps. This happens with probability $p\down$. One therefore gets
    \begin{align}
        m(t) =& \frac{\tau-t+1}{1/p\down-1 + \tau} \times 1 + \frac{t-1}{1/p\down-1 + \tau} \times p\down \notag \\
        =& \frac{\tau}{1/p\down-1 + \tau} \left( 1- \frac{(t-1)(1-p\down)}{\tau}\right).
    \end{align}
\end{proof}
We now use the short notation $+$ or $-$ if the respective system is functional or not functional in a given time step and $\ast$ if the state of the system is arbitrary or not known. The probability that a connection is broken in $3$ consecutive time steps would then be denoted by $m(3) = p_c^{- - -}$ and the effective probability for $1$ time step by $p_{\text{eff}}^b = m(1) = p_c^-$. For a single connection all probabilities $p_c^{i_1,\dots,i_t}$ for the behavior $i_1,\dots,i_t$ in $t$ consecutive time steps can be calculated in the same fashion as above by counting all possible states of the connection. One last example for this type of calculation, the probability $p_c^{- + -}$, is given below, since this probability is needed in the next subsection treating blocks of $N$ connections.
\begin{proposition}
    The probability for a single connection with failure probability $p\down$ and repairing time $\tau$ to be broken, functional and broken again in $3$ arbitrarily chosen consecutive time steps is given by
    \begin{align}
        p_c^{- + -} = \frac{(1-p\down) p\down}{1/p\down - 1 + \tau} .
    \end{align}
\end{proposition}
\begin{proof}
    The connection is on average functional for $1/p\down-1$ time steps and then broken for another $\tau$. For the connection to be broken, functional and broken again the first of the three consecutive time steps has to coincide with the last time step of the repairing time (see figure \ref{fig:prob_device_bwb}). This happens with probability $\frac{1}{1/p\down - 1 + \tau}$. After that, the connection stays functional with probability $1-p\down$ and breaks again with probability $p\down$. All in all it holds
    \begin{align}
        p_c^{- + -} = \frac{1}{1/p\down - 1 + \tau}\times (1-p\down) \times p\down .
    \end{align}
\end{proof}
\begin{figure}
    \centering
    \includegraphics[width=0.9\columnwidth]{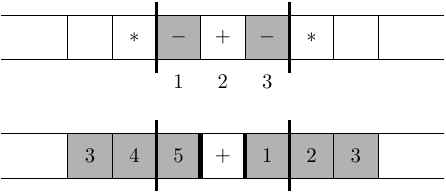}
    \caption{Possible states of a connection with repairing time $\tau=5$ which is broken, functional and broken again in $t=3$ arbitrarily chosen consecutive time steps; the first of the consecutive time steps has to coincide with the last of the broken steps of the connection; after being repaired the connection has to stay functional and then break again.}
    \label{fig:prob_device_bwb}
\end{figure}
Note that probabilities of the type $p_c^{- \ast -}$ are given by the sum of all possible cases:
\begin{align}
    p_c^{- \ast -} = p_c^{- + -} + p_c^{- - -}.
\end{align}

\subsection{A block of $N$ connections}

To calculate now the probabilities for a block of $N$ connections to be in a certain state in $t$ consecutive time steps we use two main ingredients. First, we notice that a block is broken iff every single connection is broken. So all probabilities for combinations of $-$ and $\ast$ can be directly calculated by raising the respective probability for a single connection to the power of $N$. For example it holds
\begin{align}
    p^{- - -} = \left( p_c^{- - -} \right)^N \quad \mathrm{and} \quad
    p^{- \ast -} = \left( p_c^{-\ast -} \right)^N .
\end{align}
Secondly, we can start computing the probabilities for $t=1$ arbitrary time step and then recursively calculate the still missing probabilities for more time steps (the ones which contain at least one $+$) by using the marginals. For example it holds $p^{+-} + p^{--} = p^{\ast -} = p^{-}$. For $1$ time step we get
\begin{align}
    p^{+} &= 1- p^{-},
\end{align}
for $t=2$ time steps
\begin{align}
    p^{+-} = p^{-+} &= p^{-} - p^{--},\\
    p^{++} &= p^{+} - p^{+-},
\end{align}
and for $t=3$ time steps
\begin{align}
    p^{+--} = p^{--+} &= p^{--} - p^{---},\\
    p^{-+-} &= p^{- \ast -} - p^{---},\\
    p^{-++} = p^{++-} &= p^{-+} - p^{-+-}, \\
    p^{+-+} &= p^{+-} - p^{+--}, \\
    p^{+++} &= p^{++} - p^{++-} .
\end{align}

\subsection{A chain or network}

To calculate the probability for a chain or a network to be in a given state in $t$ consecutive time steps we use the approach explained in Appendix \ref{ch:repairing}. If the functionality of the system is described by the random variable $X$ we include the time dependence by describing the system in the time steps $1,\dots,t$ by the random variables $X_1,\dots,X_t$. The system is then functional in time step $k$, iff $\id_{X_k} = 1$. The system being functional in the first step, in an arbitrary state in the second step and being broken in the third step can, e.g., be expressed by the indicator function
\begin{align}
    \id_S^{+ \ast -} = \id_{X_1} \times 1 \times (1-\id_{X_3})
\end{align}
and the according probability can be calculated using
\begin{align}
    p_S^{+ \ast -} &= \pr{X_1 = 1, X_2\in\{ 0,1 \},  X_3 = 0} \notag \\
    &= \mathbb{E}[ \id_{X_1} \times 1 \times (1-\id_{X_3}) ] 
\end{align}
If $\id_{X}$ is an expression containing different indicator functions, then each of these needs to get indexed by the time steps. So, e.g., two components in series in time step $k$ are described by $\id_{X_k} = \id_{A_k}\times \id_{B_k}$. The only thing one has to keep in mind is that different components (e.g., $A$ and $B$) may be independent, but the same component in different time steps $A_k$ and $A_{j}$, $j\neq k$, are not. We have to use the probabilities for a single block described in the section before. For example it holds for a block $A$
\begin{align}
    \mathbb{E}[ \id_{A_1} \id_{A_3} ] = p^{+ \ast +} \neq p^{+} p^{+} = \mathbb{E}[ \id_{A_1}] \mathbb{E}[\id_{A_3} ] 
\end{align}

\section{Description of the model and the numerical calculations}

In the main text of the paper, we ask whether technical devices are 
working or not, that is, whether we could in principle establish an 
entangled connection between two nodes or not. As an application, 
we simulated the actual generation of entangled links on such not 
perfectly working repeater chains and networks. In this simulation, 
to establish an entangled link between two far away parties, we 
first create shorter entangled links between repeater stations and 
then connect these by Bell measurements. In our model, we assume 
deterministic entanglement swapping. Therefore it is sufficient to 
create one link on each segment on the way between the two end nodes. 
We further assume that established links experience decoherence while 
waiting for other links to get established, this decoherence motivates 
a potential cutoff after certain time, as the entanglement in the 
links becomes too weak or disappears. We are interested in 
the secret-key rates achievable on a chain or network depending 
on the functionality of the technical devices. In the following, 
we discuss details about the simulation program (written in Python) 
and chosen simulation parameters.

We compute the waiting time as follows: We first simulate the waiting 
time for each individual connection using Monte Carlo simulation, 
see for example Ref.~\cite{Walter_2015_montecarlo}.  In the case of 
multiplexing, we determine the waiting time of a block by taking 
the minimum waiting time of all connections associated to that 
block. The waiting time of the full chain is the maximal waiting 
time among all blocks. In the network case, the waiting time of 
the full network is the waiting time of the path which gets 
established first.

We model memory decoherence as depolarisation noise as follows.
Each entangled pair is modelled as a Werner state
\begin{align}
\rho(w) = w |{\Psi^-}\rangle \langle {\Psi^-}|  + (1 - w) \frac{ \id}{4},
\end{align}
where $w$ is the visibility (or Werner parameter) and $ |{\Psi^-}\rangle$ 
is a perfect Bell state. We assume that the sources distribute maximally 
entangled states, i.e.,\ $w = 1$, but stored entangled states experience 
decoherence. At time $t$ after establishing the entangled link, the 
visibility becomes  $w(t) =  e^{-t / T_{\text{coh}}}$, where 
$T_{\text{coh}}$ is the coherence time of the memories. The  
fidelity $F = \expval{\rho(w)}{\Psi^-}$ is given by $F = \frac{ (1 + 3 w)}{4}$. 
An entanglement swapping operation using two Werner states 
$\rho(w_A), \rho(w_B)$ yields a Werner state $\rho(w_A \times w_B)$.

For the chosen simulation parameters we assume the distance from the source 
to the repeater to be $L = 100$ km. For attenuation losses of $0.2$dB/km, 
i.e., an attenuation length $L_{\text{att}} \approx 22$ km, the probability 
to generate a link successfully is given by $P_{\text{gen}} = e^{- \frac{L}
{L_{\text{att}}}} \approx 0.01$. The length of one time step is limited 
by $t_{\text{ts}} = \frac{2}{c} L \approx \frac{2}{3} 10^{-3} s$, where $c$ 
is the speed of light. We assume a memory coherence time of 
$T_{\text{coh}} = 1s = 1500 \times t_{\text{ts}}$. For a motivation of 
the simulation parameters, see Ref.~\cite{avis2023asymmetric}.

We simulate the behaviour of a network and a chain. We consider failure 
of devices on a different timescale (hours or days) than link generation 
(seconds). Therefore we assume that for the time of one entanglement 
generation attempt the network or the chain has a fixed configuration 
of working and non-working devices. The considered network is shown in 
\cref{fig:simple} and no multiplexing is used. Due to the aging process, 
some links might not work, so that one 
cannot establish entanglement with them. Using symmetries, we need to 
consider five different network configurations which are shown 
in \cref{fig:paths}.

\begin{figure*}[t]
\includegraphics[width=0.9\linewidth]{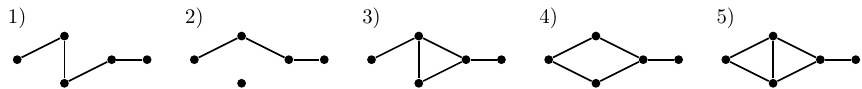}
\caption{Possible network configurations allowing for entanglement distribution
between the end nodes for the network in the Netherlands (see \cref{fig:simple}),  
up to symmetries.
\label{fig:paths}
}
\end{figure*}

The chain is of the type shown in \cref{fig:repeaterchain}. It  consist of 
$M=6$ segments with a multiplicity of $N=3$. Due to failures, we might 
establish entanglement on segments with lower multiplicity. Taking symmetries 
into account, we have 28 different functional configurations of multiplicities. 
This number results from taking all combinations of $1$, $2$ and $3$ functional 
connections per edge up to symmetry. For example, the first configurations would 
be $(1,1,1,1,1,1)$, $(1,1,1,1,1,2)$, $(1,1,1,1,1,3)$ and $(1,1,1,1,2,2)$, where 
the first three all have one functional connection in the first $5$ edges and 
$1$, $2$, $3$ connection(s) in the last edge, respectively. The last configuration 
has one connection in the first four edges and two connections in the last two edges. Note 
that the second configuration is symmetric to every permutation $(1,1,1,1,2,1)$, $(1,1,1,2,1,1)$ etc.~of it.

We modeled $6 \times 10^6$ data points for each configuration of the network. Here,
each data point is generated by a simulation of the process of entanglement generation until success was achieved. Consequently, each data point consists of a time stamp
and a fidelity of the distributed state. From the obtained data we 
calculate the secret-key rate $R= r/ \langle T \rangle$ of the BB84 protocol 
which is computed as the fraction of the secret-key fraction $r$ and the average waiting time $\langle T \rangle$. We give here a short overview, for a more 
detailed discussion see Ref. \cite{li2021}. The secret-key 
fraction for a Werner state with visibility $w$ is given by
\begin{align}
&r(w) = \\ 
&\max\lbrace 0, 1 + (1-w)\log_2 \frac{1-w}{2} + (1+w)\log_2 \frac{1+w}{2}  \rbrace. \notag
\end{align}

In order to characterize the effect of different cut-off times, we consider
the case that the entanglement generation attempt is terminated after
$t_{\mathrm{cut}}$ time steps. If there is, after this time, not in every 
edge an entangled link, then every existing link is erased and a new 
generation attempt is started. Let $T$ denote the time when every entangled 
link exists, so that the protocol reached a successful entanglement generation 
attempt. Then 
\begin{align*}
p_{\mathrm{cut}} = \pr{T\leq t_{\mathrm{cut}}}
\end{align*}
denotes the probability that the entanglement generation was successful during $t_{\mathrm{cut}}$ and \begin{align*}
    \frac{ \sum_{t=1}^{t_{\mathrm{cut}}} t \pr{T=t} }{ p_{\mathrm{cut}} } 
\end{align*}
is the average time until success during that attempt. The average time 
for which the cut-off procedure is repeated until the first successful 
attempt is reached is then geometrically distributed with $p_{\mathrm{cut}}$ 
and given by
\begin{align*}
t_{\mathrm{cut}} \sum_{k=1}^\infty k\times p_{\mathrm{cut}} ( 1-p_{\mathrm{cut}} )^k = t_{\mathrm{cut}} \frac{1-p_{\mathrm{cut}}}{p_{\mathrm{cut}}} .
\end{align*}
All in all, the average waiting time until entanglement generation is given 
by the sum of the time until the first successful attempt and the waiting 
time for that success in that attempt:
\begin{align*}
    \langle T \rangle = t_{\mathrm{cut}} \frac{1-p_{\mathrm{cut}}}{p_{\mathrm{cut}}} + \frac{ \sum_{t=1}^{t_{\mathrm{cut}}} t \pr{T=t} }{ p_{\mathrm{cut}} } .
\end{align*}
The resulting end-to-end Werner state has then a visibility parameter $W$
which is the product of the visibilities of the single-connection Werner states
in the connected path. Its mean value is given by
\begin{align*}
    \langle W \rangle = \frac{ \sum_{t=1}^{t_{\mathrm{cut}}} W(t) \pr{T=t} }{ p_{\mathrm{cut}} } 
\end{align*}

From the obtained data we can now calculate the secret-key rates $R_i$ for 
every possible functional configuration $i$ of the network and the chain 
and every choice of $t_{\mathrm{cut}}$. The average secret-key rate which 
can be achieved by this protocol on the respective topology is then given 
by the weighted sum of the different rates where the weights are the 
probabilities $\pr{\mathrm{configuration}\ i }$ of the different configurations.

The simulation is single threaded but multiple simulations were executed 
in parallel. The result was obtained by running  eight parallel simulations 
for one to two days on a computer with the following specification: Intel 
Core i7-4790 CPU with 3.60GHz running at 4 cores / 8 threads with 16 GB of 
memory available. The computation was not memory intensive. The code can be 
found online \cite{code_git}. The computed data can be made available upon
reasonable request.

\vspace{0.5cm}
\twocolumngrid
\bibliographystyle{apsrev4-2}
\bibliography{literature}

\end{document}